%% file: rangesearch.tex
\documentclass[lotsofwhite,charterfonts]{patmorin}
\usepackage{amsopn,graphicx}
\input{pat}

\DeclareMathOperator{\lft}{left}
\DeclareMathOperator{\rght}{right}
\DeclareMathOperator{\prnt}{parent}

\newcommand{\depth}{d}

\title{\MakeUppercase{Biased Range Trees}}
\author{Vida Dujmovi\'c
	\and John Howat
	\and Pat Morin}

\begin{document}
\maketitle
\begin{abstract}
A data structure, called a \emph{biased range tree}, is presented that
preprocesses a set $S$ of $n$ points in $\R^2$ and a query
distribution $D$ for 2-sided orthogonal range counting queries.  The
expected query time for this data structure, when queries are drawn
according to $D$, matches, to within a constant factor, that of the
optimal decision tree for $S$ and $D$.   The memory and preprocessing
requirements of the data structure are $O(n\log n)$.
\end{abstract}

\section{Introduction}

Let $S$ be a set of $n$ points in $\R^2$ and let $D$ be a probability
measure over $\R^2$.  A \emph{2-sided orthogonal range counting query}
over $S$ asks, for a query point $q=(q_x,q_y)$, to report the number
of points $(p_x,p_y)\in S$ such that $p_x \ge q_x$ and $p_y \ge q_y$.
A 2-sided range counting query \emph{has distribution $D$} if the
query point $q$ is chosen from the probability measure $D$.  If $T$ is
a data structure for answering 2-sided range counting queries over $S$
then we denote by $\mu_D(T)$ the expected time, using $T$, to answer a
range query with distribution $D$.  The current paper is concerned
with preprocessing the pair $(S,D)$ to build a data structure $T$ that
minimizes $\mu_D(T)$.

\subsection{Previous Work}

The general topic of geometric range queries is a field that has seen
an enormous amount of activity in the last century.  Results in this
field depend heavily on the types of objects the data structure stores
and on the shape of the query ranges.  In this section we only mention
a few data structures for orthogonal range counting and semigroup
queries in 2 dimensions.  The interested reader is directed to the
excellent, and easily accessible, survey by Agarwal and Erickson
\cite{ea99}.

Orthogonal range counting is a classic problem in computational
geometry.  The 2- (and 3- and 4-) sided range counting problem can be
solved by Bentley's \emph{range trees} \cite{ae42}.  Range trees use
$O(n\log n)$ space and can be constructed in $O(n\log n)$ time.
Originally, range trees answered queries in $O(\log^2 n)$ time.
However, with the application of fractional cascading
\cite{ae76,ae196} the query time can be reduced to $O(\log n)$ without
increasing the space requirement by more than a constant factor.
Range trees can also answer more general \emph{semigroup queries} in
which each point of $S$ is assigned a weight from a commutative
semigroup and the goal is to report the weight of all points in the
query range \cite{ae133,ae292}.

For 2-sided orthogonal range counting queries, Chazelle
\cite{ae55,ae58} proposes a data structure of size $O(n)$, that can be
constructed in $O(n\log n)$ time, and that can answer range couting
queries in $O(\log n)$ time.  Unfortunately, this data structure is
not capable of answering semigroup queries in the same time bound.
For semigroup queries, Chazelle provides data structures with the
following requirements: (1)~$O(n)$ space and $O(\log^{2+\epsilon} n)$
query time, (2)~$O(n\log\log n)$ space and $O(\log^{2}n\log\log n)$
query time, and (3)~$O(n\log^\epsilon n)$ space and $O(\log^2 n)$
query time.

Practical linear space data structures for range counting include
$k$-d trees \cite{ae41}, quad-trees \cite{ae251}, and their
variants.  These structures are practical in the sense that they are
easy to implement and use only $O(n)$ space.  Unfortunately, neither
of these structures has a worst-case query time of $\log^{O(1)} n$.
Thus, in terms of query time, $k$-d trees and quad-trees are nowhere
near competitive with range trees.


Despite the long history of data structures for orthogonal range
queries, range trees with fractional cascading are still the most
effective data structure for 2-sided orthogonal range queries in the
semigroup model.  In particular, no data structure is currently known
that uses $o(n\log n)$ space and can answer 2-sided orthogonal range
queries in $O(\log n)$ time. 

\subsection{New Results}

In the current paper we present a data structure, the \emph{biased
range tree}, for 2-sided orthogonal range counting.  Biased range
trees fit into the \emph{comparison tree} model of computation, in
which all decisions made during a query are based on the result of
comparing either the $x$- or $y$-coordinate of the query point to some
precomputed values.  Most data structures for orthogonal range
searching, including range trees, $k$-d trees and quadtrees, fit into
the comparison tree model. This model makes no assumptions about the
$x$- or $y$-coordinates of points other than that they each come from
some (possibly different) total order.  This is particularly useful in
practice since it avoid the precision problems usually associated with
algebraic decisions and allows the mixing of different data types
(one for $x$-coordinates and one for $y$-coordinates) in one data
structure.

A biased range tree has size $O(n\log n)$, can be constructed in
$O(n\log n)$ time, and can answer range counting (or semigroup)
queries in $O(\mu_D(T^*))$ expected time, where $T^*$ is any
comparison tree that answers range counting queries over $S$.  In
particular, $T^*$ could be a comparison tree that minimizes
$\mu_D(T^*)$ implying that the expected query time of our data
structure is as fast as the fastest comparison-based data structure
for answering range counting queries over $S$.  Moreover, the
worst-case search time of biased range trees is $O(\log n)$, matching
the worst-case performance of range trees.

Note that we do not place any restrictions on the comparison tree
$T^*$.  Biased range trees, while requiring only $O(n\log n)$ space,
are competitive with any comparison-based data structure.  Thus, the
memory requirement of biased range trees is the same as that of range
trees but their expected query time can never be any worse.
 
The remainder of the paper is organized as follows. In
\secref{preliminaries} we present background material that is used in
subsequent sections.  In \secref{data-structure} we define biased
range trees. In \secref{lower-bound} we prove that biased range trees
are optimal.  In \secref{summary} we recap, summarize, and describe
directions for future work.

\section{Preliminaries}
\seclabel{preliminaries}

In this section we give definitions, notations, and background
that are prerequisites for subsequent sections.

\paragraph{Rectangles.}

For the purposes of the current paper, a \emph{rectangle}
$R(a,b,c,d)$ is defined as
\[
    R(a,b,c,d) = \{ (x,y) : \mbox{$a\le  x \le b$ and $c \le y \le d$}\}
	\enspace .
\]
We also allow unbounded rectangles by setting $a,c=-\infty$ and/or
$b,d=\infty$.  Therefore, under this definition, rectangles can have
0, 1, 2, 3, or 4 sides.  For a query point $q=(q_x,q_y)$ we denote 
by $R(q)$ the query range $R(q_x,\infty,q_y,\infty)$.  A
\emph{horizontal strip} is rectangle of the form
$R(-\infty,\infty,c,d)$ and a \emph{vertical strip} is a rectangle of
the form $R(a,b,-\infty,\infty)$.

\paragraph{Classification Problems and Classification Trees.}

A \emph{classification problem} over a domain $\mathcal{D}$ is a
function $\mathcal{P}:\mathcal{D}\mapsto \{0,\ldots,k-1\}$.  The
special case in which $k=2$ is called a \emph{decision problem}.  A
$d$-ary \emph{classification tree} is a full $d$-ary tree\footnote{A
full $d$-ary tree is a rooted ordered tree in which each non-leaf node
has exactly $d$ children.} in which each internal node $v$ is labelled
with a function $P_v:\mathcal{D}\mapsto\{0,.\ldots,d-1\}$ and for
which each leaf $\ell$ is labelled with a value
in $\{0,\ldots,k-1\}$. The \emph{search path} of an input $q$
in a classification tree $T$ starts at the root of $T$ and, at each
internal node $v$, evaluates $i=P_v(q)$ and proceeds to the $i$th
child of $v$.  We denote by $T(q)$ the label of the final (leaf) node
in the search path for $q$.  We say that the classification tree $T$
\emph{solves} the classification problem $\mathcal{P}$ over the domain
$\mathcal{D}$ if, for every $q\in \mathcal{D}$, $\mathcal{P}(q)=T(q)$.

The particular type of classification trees we are concerned with are
\emph{comparison trees}.  These are binary classification trees in
which the function $P_v$ at each node $v$ compares either $q_x$ or
$q_y$ to a fixed value (that may depend on the point set $S$ and the
distribution $D$).  For the problem of 2-sided range counting over
$S$, the leaves of $T$ are labelled with values in $\{0,\ldots,|S|\}$
and $T(q)=|R(q)\cap S|$ for all $q\in\R^2$.

\paragraph{Probability.}

For a probability measure $D$ and an event $X$, we denote by $D_{|X}$ the
distribution $D$ conditioned on $X$.  That is, the distribution where
the probability of an event $Y$ is $\Pr(Y\mid X)=\Pr(Y\cap X)/\Pr(X)$.
The probability measures used in this paper are usually defined over
$\R^2$.  We make no assumptions about how these measures are
represented, but we assume that an algorithm can, in constant time,
given a rectangle $r$, determine $\Pr(r)$.

For a classification tree $T$ that solves a problem
$P:\mathcal{D}\mapsto\{0,\ldots,k-1\}$ and a probability measure $D$
over $\mathcal{D}$, the \emph{expected search time} of $T$, denoted
by $\mu_D(T)$, is the
expected length of the search path for $q$ when $q$ is drawn at random
from $\mathcal{D}$ according to $D$.  Note that, for each leaf $\ell$
of $T$ there is a maximal subset $r(\ell)\subseteq \mathcal{D}$ such
that the search path for any $q\in r(\ell)$ ends at $\ell$.  Thus, the
expected search time of $T$ (under distribution $D$) can be written as
\[
     \mu_D(T) = \sum_{\ell\in L(T)} \Pr(r(\ell))\times \depth_T(\ell)
	\enspace ,
\]
where $L(T)$ denotes the leaves of $T$ and $\depth_T(\ell)$ denotes the
length of the path from the root of $T$ to $\ell$.  When the tree $T$
is obvious based on context we will sometimes use the notation
$d(\ell)$ to denote $d_T(\ell)$. Note that, for
comparison trees, the closure of $r(\ell)$ is always a rectangle.  For
a node $v$ in a tree, we will use the phrases \emph{depth of $v$} and
\emph{level of $v$} interchangeably and they both refer to $d(v)$. 

The following theorem is a restatement of (half of) Shannon's
Fundamental Theorem for a Noiseless Channel \cite[Theorem 9]{s48}.
\begin{thm}\thmlabel{shannon}
Let $\mathcal{P}:\mathcal{D}\mapsto \{0,\ldots,k-1\}$ be a classification
problem and let $p\in \mathcal{D}$ be selected from a distibution $D$ such
that $\Pr\{\mathcal{P}(p)= i\}=p_i$, for $0\le i< k$.  Then, any
$d$-ary classification tree $T$ that solves $\mathcal{P}$ has
\begin{equation}
     \mu_D(T) \ge \sum_{i=0}^{k-1} p_i\log_d(1/p_i) \enspace .
	\eqlabel{shannon}
\end{equation}
\end{thm}

In terms of range counting, \thmref{shannon} immediately implies that,
if $p_i$ is the probability that the query range contains $i$ points
of $S$, then any binary decision tree $T$ that does range counting has
$\mu_D(T) \ge \sum_{i=0}^{n} p_i\log(1/p_i)$.  Unfortunately for us,
this lower bound is too weak and, in general, there is no decision
tree whose performance matches this obvious entropy lower bound.

A stronger lower bound on the cost of range searching can be obtained
by considering the arrangement $A$ of $2n$ rays obtained by drawing
two rays originating at each point of $S$, one to the left and one
downwards (see \figref{arrangement}.a).  This arrangement partitions
the plane into a set of faces $F(A)$.  If $T$ is a comparison tree for
range counting in $S$, then there is no leaf $\ell$ of $T$ such that
the interior of $r(\ell)$ intersects any edge of $A$ since otherwise
there are query points $q$ in the neighbourhood of this intersection
for which $T(q)\neq |R(q)\cap S|$.  Therefore, by relabelling the leaves
of $T$ with the faces of $A$, we obtain a data structure for
determining which face of $A$ contains the query point $q$.
By \thmref{shannon}, this implies that
\[
   \mu_D(T) \ge \sum_{f\in F(A)} \Pr(f)\log(1/\Pr(f)) \enspace .
\]
Unfortunately, this bound is still not strong enough and, in general,
there is no decision tree $T$ that matches this lower bound.  To see
this, consider \figref{arrangement}.b, when the query point $q$ is
uniformly distributed among the $n+1$ shaded circles.  In this case,
$q$ is always in the same face of $A$ so the lower bound given above
is 0.  Nevertheless, it is not hard to see that the leaves of
any decision tree $T$ for range searching in $S$ can be relabelled to
determine which of the $n+1$ circles contains $q$, so $\mu_D(T) \ge
\log(n+1)$.

\begin{figure}
  \begin{center}
    \begin{tabular}{cc}
      \includegraphics{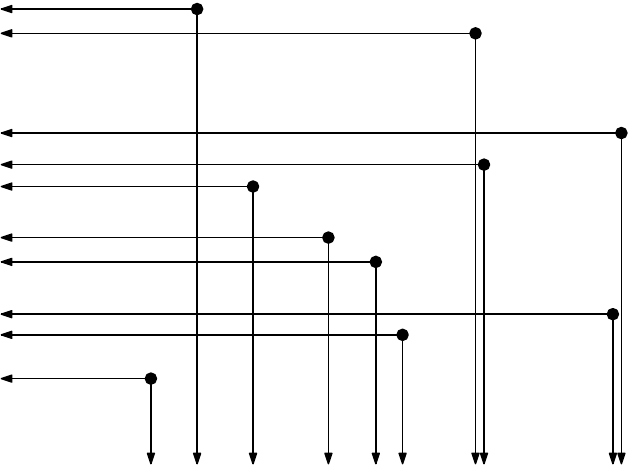} &
      \includegraphics{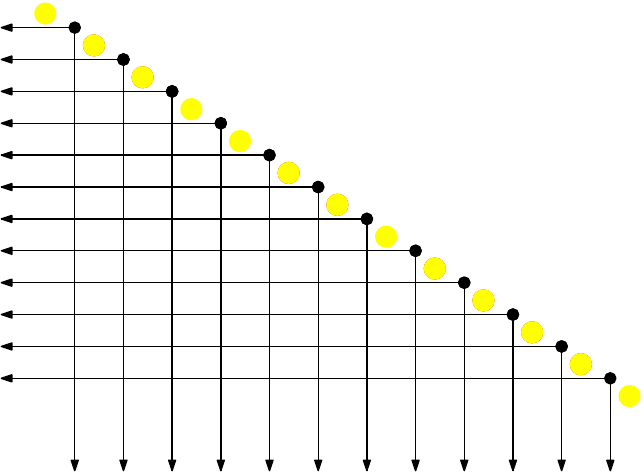} \\
        (a) & (b)
    \end{tabular}
  \end{center}
  \caption{(a)~The distribution of the query point $q$ over the faces
     of the arrangement $A$ gives a lower bound on the cost of any
     comparison tree for range counting in $S$. (b)~The lower bound is
     not always achievable by a comparison tree.}
  \figlabel{arrangement}
\end{figure}

\paragraph{Biased Search Trees.}

\emph{Biased search trees} are a classic data structure for solving
the following 1-dimensional problem:  Given an increasing sequence of
real numbers $X=\langle x_0=-\infty,x_1,x_2,\ldots ,
x_n,x_{n+1}=\infty\rangle$ and a probability distribution $D$ over
$\R$, construct a binary search tree  $T=T(X,D)$ so that, for any
query value $q$ drawn from $D$, one can quickly find the unique
interval $[x_i,x_{i+1})$ containing $q$.  If $p_i$ is the probability
that $q\in[x_i,x_{i+1})$ then the expected number of comparisons
performed while searching for $q$ is given
by
\[
   \mu_D(T) \le \sum_{i=1}^{n} p_i\log(1/p_i) + 1 
\]
and the tree $T$ can be constructed in $O(n)$ time \cite{m75}.
Clearly, by \thmref{shannon}, the query time of this binary search
tree is optimal up to an additive constant term.  Note that, by having
each node of $T$ store the size of its subtree, a biased search tree
can count the number of elements of $X$ in the interval
$I(q)=[q,\infty)$ without increasing the search time by more than a
constant factor.  Thus, biased search trees are an optimal data
structure for 1-dimensional range counting.

\section{Biased Range Trees}
\seclabel{data-structure}

In this section we describe the biased range tree data structure,
which has three main parts: the backup tree, the primary tree, and a
set of catalogues that adorn the nodes of the primary tree.

\subsection{The Backup Tree}

In trying to achieve optimal query time, biased range trees will try
to quickly answer queries that are, in some sense, easy.  In some
cases, a query is difficult and it cannot be answered in $o(\log n)$
time.  For these queries, a \emph{backup} range tree that stores the
points of $S$ and can answer any 2-sided range query in $O(\log n)$
worst-case time is used.  The preprocessing time and space
requirements of this backup tree are $O(n\log n)$ \cite{bkos97}.

\subsection{The Primary Tree}

Like a range tree, a biased range tree is an augmented data structure
consisting of a primary tree whose nodes store secondary structures.
However, in a range tree the primary tree is a binary search tree that
discriminates based only on the $x$-coordinate of the query point $q$.
In order to achieve optimal expected query time, this turns out to be
insufficient, so instead biased range trees use a variation of a $k$-d
tree as the primary tree.

The primary tree is constructed in a top-down fashion.  Each node $v$
of $T$ is associated with a region $r(v)$ whose closure is a
rectangle.  The region associated with the root of $T$ is all of
$\R^2$.  We say that a node $v$ is \emph{bad} if its depth is at least
$\ceil{\log_2 n}$ and $r(v)\cap S \neq \emptyset$.  A node $v$ is
\emph{split} if $v$ its depth is less than $\ceil{\log_2 n}$, and
$r(v)\cap S\neq \emptyset$.  The two children of a split node $v$ are
associated with the two regions obtained by removing a horizontal or
vertical strip $s(v)$ from $r(v)$ depending on whether the depth of
$v$ is even or odd, respectively.  We call a node $v$ at even distance
from the root a \emph{vertical node}, otherwise we call $v$ a
\emph{horizontal node}. 

Refer to \figref{left-right}.  For a vertical node $v$, we denote its
children by $\lft(v)$ and $\rght(v)$ and call them the \emph{left
child} and \emph{right child} of $v$, depending on which side of the
vertical strip (left or right) they are.  For uniformity, we will also
call the children of a node $v$ that is split with a horizontal strip
$\lft(v)$ and $\rght(v)$.  The child below the strip is denote by $\lft(v)$ and
the child above the strip is denoted by $\rght(v)$.
Similarly, the left and right
boundaries of a strip $s(v)$ at a horizontal node $v$ refer to the
bottom and top sides of $s(v)$.  Note that, with these conventions, if
the query point $q$ is in $r(\lft(v))$ then $R(q)$ intersects
$r(\rght(v))$.  However, if $q\in r(\rght(v))$ then $R(q)$ does not
intersect $r(\lft(v))$.  Similarly, for a query point $q\in s(v)$, the
query range $R(q)$ intersects $r(\rght(v))$ but not $r(\lft(v))$

\begin{figure}
  \begin{center}
    \begin{tabular}{cc}
      \includegraphics{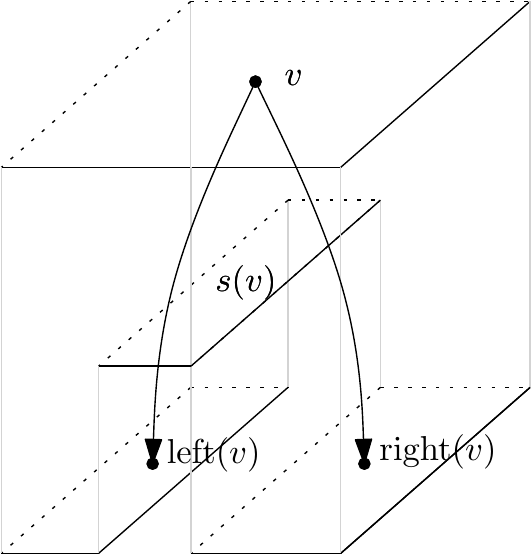} & \includegraphics{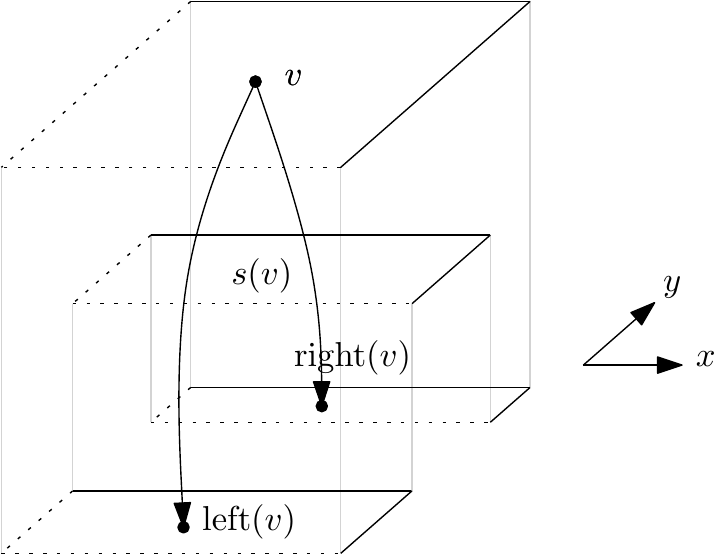} \\
      (a) & (b)
    \end{tabular}
  \end{center}
  \caption{The splitting of (a)~a vertical node $v$ and (b)~a horizontal
  node $v$.}
  \figlabel{left-right}
\end{figure}

All that remains is to define the strip $s(v)$ for each node $v$. If
$v$ is a leaf then we use the convention that $s(v)=r(v)$.  If $v$ is
not a leaf then $s(v)\subseteq r(v)$ is selected as a maximal strip
containing no point of $r(v)\cap S$ in its interior, that is closed on
its right side and open on its left side and such that each of the at
most two components of $r(v)\setminus s(v)$ has probability at most
$\Pr(r(v))/2$.  Suppose $v$ is a vertical node.  Then let
$r(v)_1,\ldots,r(v)_k$, be a partitioning of $r(v)$ into strips, in
left-to-right order, obtained by drawing a vertical line through each
of the $k$ points in $S\cap r(v)$.  We use the convention that each
strip is closed on its right side and open on its left side.  Then
there is a unique strip $s(v)=r(v)_i$ such that $\sum_{j=1}^{i-1}
\Pr(r(v)_j) \le \Pr(r(v))/2$ and $\sum_{j=i+1}^{k} \Pr(r(v)_j) <
\Pr(r(v))/2$.  For a
horizontal node $v$, the definition of $s(v)$ is analagous except we
use horizontal lines through each point of $r(v)\cap S$.

Note that for a node $v$ that is not a leaf, we use the convention
that $s(v)$ contains its right side but not its left side and that
$r(\rght(v)$ and $r(\lft(v))$ are the two components of $r(v)\setminus
s(v)$.  This implies that $r(\lft(v))$ and/or $r(\rght(v))$ may be
empty, in which case $\lft(v)$, respectively, $\rght(v)$ is a leaf of
$T$.  With these definitions, for any point $q\in \R^2$ there is
exactly one vertex $v$ of $T$ such that $q\in s(v)$.

The following two properties are easily derived from the definition of
$T$ and are necessary to prove the optimality of biased range trees:
\begin{enumerate}

\item Any node $v$ at depth $i$ in $T$ has $\Pr(s(v))\le \Pr(r(v))\le 1/2^i$.

\item For any node $v$ of $T$, if $\Pr(r(v)) > 0$, then the closure of
$r(v)$ contains at least one point of $S$.
\end{enumerate}

Point~1 above follows immediately from the definition of $s(v)$. Next
we explain the logic leading to Point~2.  If $r(v)$ contains a point
of $S$ then so does the closure of $r(v)$.   If $r(v)=\emptyset$, then
$\Pr(r(v)) = 0$.  Otherwise, $r(v) \neq \emptyset$ and $r(v)$ has no
point of $S$ in its interior.  Then consider the parent $w$ of $v$.
Since $s(w)$ does not contain $r(v)$ there must be a point of $S$ on
the boundary of $s(w)$ that is also on the boundary of $r(v)$.
Therefore $r(v)$ contains this point in its closure.

\subsection{The Catalogues}
\seclabel{catalogues}

The nodes of the tree $T$ are augmented with additional data
structures called \emph{catalogues} that hold subsets of $S$.  Each
node $v$ has two catalogues, $C_x(v)$ and $C_y(v)$ that store subsets
of $S$ sorted by their $x$-, respectively, $y$-, coordinate.
Intuitively, $C_x(v)$ stores points that are ``above'' $r(v)$ and
$C_y(v)$ stores points that are ``to the right of'' $r(v)$.  (Refer to
\figref{catalogues}.)  More precisely, if $v$ is a horizontal node,
then $C_x(\lft(v))= (s(v)\cup r(\rght(v)))\cap S$ and
$C_y(\lft(v))=\emptyset$.  If $v$ is a vertical node, then
$C_y(\lft(v)) = (s(v)\cup r(\rght(v)))\cap S$ and
$C_x(\lft(v))=\emptyset$.  For any node $v$ that is the root of $T$ or
a right child of its parent, $C_x(v)=C_y(v)=\emptyset$.  

\begin{figure}
  \begin{center}
    \begin{tabular}{c|c}
      \includegraphics{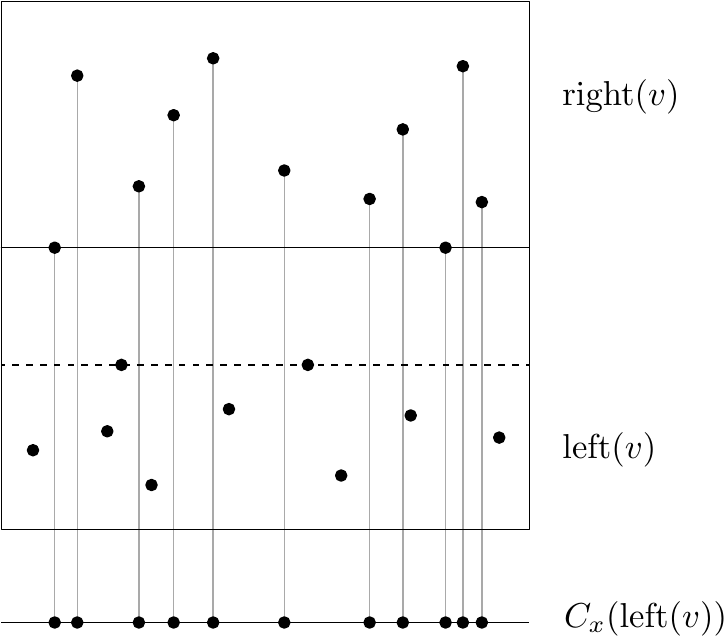} &
        \includegraphics{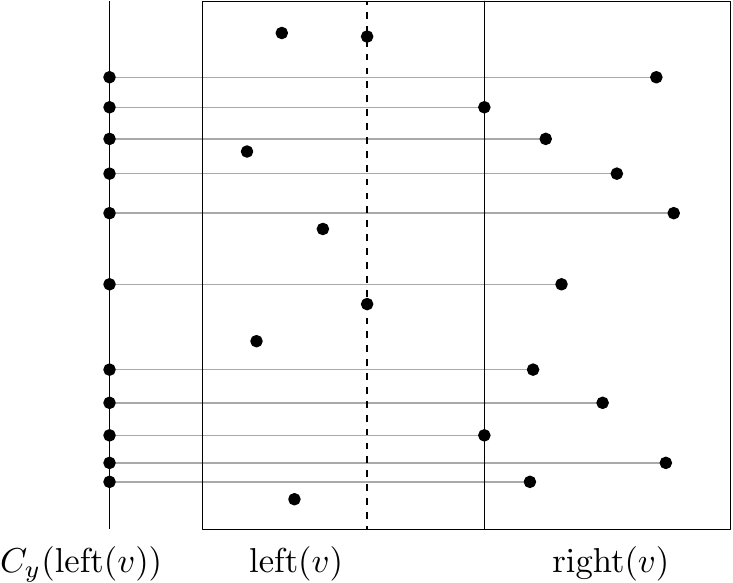} \\
       (a) & (b) 
    \end{tabular}
  \end{center}
  \caption{The catalogues of (a)~a horizontal node $v$ and (b)~a
  vertical node $v$.}
  \figlabel{catalogues}
\end{figure}

Consider any node $v$ that is not a bad leaf and any point $q\in s(v)$.
If $v$ has a left child then let $v_1=\lft(v)$, otherwise, let
$v_1=v$.
Let $v_1,\ldots,v_k$ denote the path from $v_1$ to the root of $T$
(see \figref{3d-view}).  Then the catalogues of $v_1,\ldots,v_k$ have
the following properties:
\begin{figure}
  \begin{center}
    \includegraphics{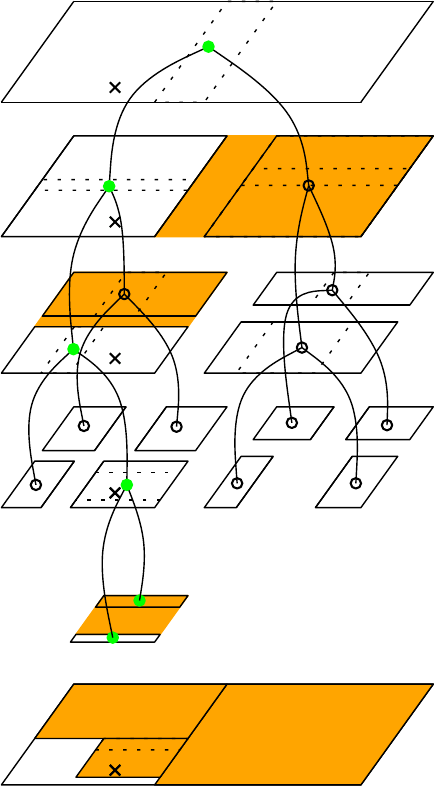}
  \end{center}
  \caption{The area covered by catalogues on the path $v$ to the root
of $T$. The $\times$ symbol shows the location of the query point $q$.}
  \figlabel{3d-view}
\end{figure}

\begin{enumerate}

\item The points in the catalogues of $v_1,\ldots,v_k$ are above or to
the right of $q$.  That is, for each $1\le i \le k$, all points in
$C_y(v_i)$, respectively, $C_x(v_i)$ have their $x$-, respectively,
$y$-, coordinate greater than or equal to $q_x$, respectively, $q_y$.

\item All catalogues at nodes in $v_1,\ldots,v_k$ are disjoint.  That,
is, for each $1\le i\le j \le k$,
$C_x(v_i)\cap C_x(v_j) = \emptyset$,
$C_y(v_i)\cap C_y(v_j) = \emptyset$,
$C_x(v_i)\cap C_y(v_j) = \emptyset$, and
$C_x(v_j)\cap C_y(v_i) = \emptyset$.

\item The catalogues at nodes $v_1,\ldots,v_k$ contain all points in
the query range $R(q)$.  That is,
\[
     R(q)\cap S \subseteq \bigcup_{i=1}^k \left(C_x(v_i)\cup C_y(v_i)\right)
        \enspace . 
\]
\end{enumerate}

Note that, points 1, 2 and 3 above imply that determining $|R(q)\cap
S|$ can be done by solving a sequence of 1-sided range queries in the
$x$- and $y$-catalogues of $v_1,\ldots,v_k$.  However, performing these queries
individually would take too long.

To speed up the process of navigating the catalogues of $T$,
fractional cascading \cite{ae76} is used.  Starting at the root of $T$ and as long as $v$ is
not a leaf, a fraction of the data in $C_x(v)$ is cascaded into
$C_x(\rght(v))$ and $C_x(\lft(v))$.  As well, a fraction of the data
in $C_y(v)$ is cascaded into both $C_y(\rght(v))$ and $C_y(\lft(v))$.
%
%
%
%
%
Note that this
cascading is done only to speed up navigation between the catalogues
of $T$.  Although fractional cascading introduces extra data into the
catalogues of $T$ we will continue to use the notations $C_x(v)$ and
$C_y(v)$ to denote the set of points contained in the catalogues of
$v$ before fractional cascading takes place.


Finally, each catalogue $C_x(v)$ and $C_y(v)$ is indexed by a biased
binary search tree $T_x(v)$, respectively, $T_y(v)$.  If $v$ is the
left child of its parent, then the weight of an interval $(a,b]$ in
$T_x(v)$, respectively, $T_y(v)$ is given by the probability that
$q_x$, respectively, $q_y$, is in the interval $(a,b]$ when $q$ is
drawn according to the distribution $D_{\mid s(\prnt(v))}$.  Otherwise
($v$ is not a left child), the weight of an interval is determined by
the distribution $D_{\mid s(v)}$.

\subsection{Construction Time and Space Requirements}

The biased range tree data structure is now completely defined.  The
structure consists of a backup tree, a primary tree, and the
catalogues of the primary tree.  We now analyze the construction time
and space requirements of biased range trees.

The backup tree has size $O(n\log n)$ and can be constructed in
$O(n\log n)$ time \cite[Theorem~5.11]{bkos97}.  To construct the
primary tree quickly we presort the points of $S$ by their $x$ and $y$
coordinates.  Since the primary tree has height $O(\log n)$, it is
then easily constructed in $O(n\log n)$ time.  Ignoring any copies of
points created by fractional cascading, each point in $S$ occurs in at
most 2 catalogues at each level of the primary tree.  Thus, the sizes
of all catalogues (before fractional cascading) is $O(n\log n)$ and
these catalogues can be constructed in $O(n\log n)$ time (because of
elements of $S$ are presorted; see de~Berg \etal\
\cite[Section~5.3]{bkos97} for details).  The fractional cascading
between catalogues does not increase the size of catalogues by more
than a constant factor since each catalogue is cascaded into only a
constant number of other catalogues \cite{ae76}.

In summary, given the point set $S$ and access to the distribution
$D$, a biased range tree for $(S,D)$ can be constructed in $O(n\log
n)$ time and requires $O(n\log n)$ space.

\subsection{The Query Algorithm}

The algorithm to answer a 2-sided range query $q=(q_x,q_y)$ proceeds
in three steps:

\begin{enumerate}

\item The algorithm navigates the tree $T$ from top to bottom to
locate the unique node $v$ such that $q\in s(v)$. This step takes
$O(d_T(q))$ time, where $d_T(q)$ is the depth of the node $v$.  If $v$
is a bad leaf (so $d_T(q)\ge \log n$) then the algorithm performs a
range query in $O(\log n)$ time using the backup range tree and the
query algorithm does not execute the next two steps.

\item  If $v$ has a left child then let $u=\lft(v)$, otherwise let
$u=v$. The algorithm uses $T_x(u)$ and $T_y(u)$ to locate $q_x$ and
$q_y$, respectively, in the catalogues $C_x(u)$ and $C_y(u)$,
respectively.

\item The algorithm walks back from $u$ to the root of $T$, locating
$q$ in the catalogues of all nodes on this path and computing the
results of the range counting query as it goes.  Thanks to fractional
cascading, each step of this walk can be done in constant time, so the
overall time for this step is also $O(d_T(q))$.
\end{enumerate}

Observe that Steps~1 and 3 of the query algorithm each take
$O(d_T(q))$ time.  The time needed to accomplish Step~2 of the
algorithm depends on exactly what is in the catalogues $C_x(u)$ and
$C_y(u)$, and will be the first quantity we study in the next section.

\section{Optimality of Biased Range Trees}
\seclabel{lower-bound}

In this section we show that the expected query time of biased range
trees is as good as the expected query time of any comparison tree.
The expected query time has two components.  The first component is
the expected depth, $d_T(q)$,  of the node $v$ such that $s(v)$
contains $q$.  The second component is the expected cost of locating
$q$ in the catalogues of $u$ (recall that $u=\lft(v)$ or $u=v$ if $v$
has no left child).  We will show that each of these two components
is a lower bound on the expected cost of any decision tree for
two-sided range searching on $S$ where queries come from distribution
$D$.  In order to simplify notation in this section we will use the
convention $\Pr(v)=\Pr(s(v))$ is the probability that a search
terminates at node $v$ of $T$.  

\subsection{The Catalogue Location Step}

First we show that the expected cost of locating $q$ in the two
catalogues, $C_x(u)$ and $C_y(u)$ is
a lower bound on the expected cost of any decision tree for answering
2-sided range queries in $S$.  The intuition behind this proof is
that, in order to correctly answer range counting queries, any decision tree
for range counting must locate the $x$-coordinate of $q$
with respect to the $x$-coordinates of all points above $q$.  
Similarly, it must locate the $y$-coordinate of $q$ with respect to
the $y$-coordinates of all points to the right of $q$.  The structure
of the catalogues ensures that biased range trees do this in the most
efficient manner possible.

\begin{lem}\lemlabel{cataloguer}
Let $S$ be a set of $n$ points and let $D$ be a probability measure
over $\R^2$.
Let $T^*$ be any decision tree for 2-sided range counting in $S$ and let
$C_2(S,D)$ denote the expected cost of locating $q$ in Step 2 of the
biased range tree query algorithm on the biased range tree $T=T(S,D)$. 
Then
\[
  \mu_D(T^*) = \Omega(C_2(S,D)) \enspace .
\] 
\end{lem}

\begin{proof}
We first observe that, by definition,
\[
  C_2(S,D) =  \sum_{v\in T} 
              \Pr(v)\left( \mu_{D_{\mid s(v)}}(T_x(u))
                               +  \mu_{D_{\mid s(v)}}(T_y(u)) \right)
           \enspace .
\]
Consider some node $v$ of $T$.  For a point $q\in s(v)$, all of the
points in $T_x(v)$ are points that may or may not be in the query
range $R(q)$ depending on where exactly $q$ is located within $s(v)$.
This implies that, if $T^*$ correctly answers range queries for every
point $q\in s(v)$ then it must determine the location of the
$x$-coordinate of $q$ with respect to all points in $T_x(v)$.  More
precisely, the leaves of $T^*$ could be relabelled to obtain a
comparison tree that determines, for any $q\in s(v)$, which interval
of $T_x(v)$ contains $q_x$.  Since $T_x(u)$ is a biased search tree
for the probability measure $D_{\mid s(v)}$, this implies that
\[
  \mu_{D_{\mid s(v)}}(T^*) \ge \mu_{D_{\mid s(v)}}(T_x(u)) - 1\enspace .
\]
Similarly, the same argument applied to $T_y(v)$ yields 
\[
  \mu_{D_{\mid s(v)}}(T^*) \ge \mu_{D_{\mid s(v)}}(T_y(u)) - 1\enspace .
\]
We can now complete the proof with
\begin{eqnarray*}
\mu_D(T^*) 
 & = & \sum_{v\in T} \Pr(v)\cdot\mu_{D_{\mid s(v)}}(T^*) \\
 & \ge & \sum_{v\in T}
	\Pr(v) \cdot\max\left\{\mu_{D_{\mid s(v)}}(T_x(u)), 
		       \mu_{D_{\mid s(v)}}(T_y(u))\right\} - 1 \\
 & \ge & \sum_{v\in T}
	\frac{1}{2}\Pr(v)\cdot\left( \mu_{D_{\mid s(v)}}(T_x(u))
                             +  \mu_{D_{\mid s(v)}}(T_y(u)) \right) - 1 \\
 & = & \frac{1}{2}\cdot C_2(S,D) - 1 = \Omega(C_2(S,D)) \enspace .
\end{eqnarray*}
\end{proof}

\subsection{The Tree Searching Step}

Next we bound the expected depth $d_T(q)$ of the node $v$ of $T$ such
that $q\in s(v)$.  We do this by showing that any decision tree $T^*$
for range counting in $S$ must solve a set of point location problems
and that the expected depth of $v$ is a lower bound on the complexity
of solving these problems.

We say that a set of rectangles is \emph{HV-independent} if no
horizontal or vertical line intersects more than one rectangle in the
set.  We say that a set $\{v_1,\ldots,v_k\}$ of nodes in $T$ is
\emph{HV-independent} if the set $\{r(v_1),\ldots,r(v_k)\}$ is
HV-independent.

\begin{lem}\lemlabel{independent}
Let $S$ be a set of $n$ points and let $D$ be a probability measure
over $\R^2$.
Let $T=T(S,D)$ be the biased range tree for $(S,D)$ and label
each node of $T$ white
or black, such that all white nodes are at distance at most $i$ from
the root of $T$.  Then, if $T$ contains more than $\gamma^i$ white nodes
then $T$ contains an HV-independent set of white nodes of size
$\Omega((\gamma/\sqrt{2})^i)$.
\end{lem}

\begin{proof}
Define a graph $G=(V,E)$ whose vertices are the white nodes of $T$ and
for which $uv\in E$ if and only if there is a horizontal or vertical line that
intersects both $r(u)$ and $r(v)$.  Note that an independent set of
vertices in $G$ is an HV-independent set of which nodes in $T$.  Thus,
it suffices to find a sufficiently large independent set in $G$

A well-know result on $k$-d trees states that, for a $k$-d tree of
height $i$, any horizontal or vertical line intersects at most
$2^{\ceil{i/2}}$ rectangles of the $k$-d tree
\cite[Lemma~5.4]{bkos97}.  Therefore, since $T$ is a $k$-d
tree,\footnote{Although $T$ is not exactly a $k$-d tree as described
in Reference~\cite{bkos97}, the proof found there still holds.} the
number of edges in $G$ is at most $|V|\cdot 2^{\ceil{i/2}}$.  This
implies that $G$ has a vertex $v$ of degree at most $2^{\ceil{i/2}+1}$
and this is also true of any vertex-induced subgraph of $G$.

We can therefore obtain an independent set in $G$ by repeatedly
selecting a vertex $v$ of degree $2^{\ceil{i/2}+1}$, adding $v$ to the
independent set and deleting $v$ and its neighbours from $G$.  Since, at
each step we add one vertex to the independent set and delete at most
$2^{\ceil{i/2}+1}+1$ vertices from $G$, this produces an independent of size
$\Omega(|V|/2^{i/2}) = \Omega((\gamma/\sqrt{2})^i)$, as required.
\end{proof}

We can now provide the second piece of the lower bound.

\begin{lem}\lemlabel{depth}
Let $S$ be a set of $n$ points and let $D$ be a probability measure
over $\R^2$.
Let $T^*$ be any comparison tree that does range counting over $S$. Let
$C_1(S,D)$ denote the expected depth of the node $v$ of the biased
range tree $T=T(S,D)$ such that $q\in s(v)$.  Then
\[
    \mu_D(T^*) = \Omega(C_1(S,D))
\]
\end{lem}

\begin{proof}
Partition the nodes of $T$ into groups $G_1,G_2,\ldots$ where $G_i$
contains all nodes $v$ such that $1/2^{i} \le \Pr(v) \le 1/2^{i-1}$.
Observe that the nodes in group $G_i$ occur in the first $i$ levels of
$T$.  Select a constants $\gamma$ and $\beta$ with $\sqrt{2} < \gamma
< \beta < 2$ and define $\alpha=\gamma/\sqrt{2}$.  By
repeatedly applying \lemref{independent}, each group $G_i$ can be
partitioned into groups $G_{i,1},\ldots,G_{i,t_i}$ where, for each $1
\le j < t_i$, $G_{i,j}$ is an HV-independent set with $|G_{i,j}|
\ge \alpha^i$.  Furthermore, $|G_{i,t_i}| \le \gamma^i$. (Note that
$G_{i,t_i}$ is not necessarily HV-independent.)

Consider some group $G_{i,j}$ for $1\le j < t_i$.  Let $\ell$ be a
leaf of $T^*$ and observe that, because the nodes in $G_{i,j}$ are
independent and each one contains at least one point of $S$ in its
closure, there are at most 4 nodes $v$ in $G_{i,j}$ such that
$r(\ell)$ intersects the closure of $r(v)$.  (Otherwise $r(\ell)$
contains a point of $S$ in its interior and therefore $T^*$ does not
solve the range counting problem for $S$.) Thus, by performing 2 
additional comparisons, $T^*$ can be used to determine which node of
$v\in G_{i,j}$ (if any) contains the query point $q$ in $s(v)$.
However, $G_{i,j}$ contains $\Omega(\alpha^i)$ nodes and the search
path for $q$ terminates at each of these with probability between
$1/2^i$ and $1/2^{i-1}$.  Therefore, if we denote by $D_{i,j}$ the
distribution $D$ conditioned on the search path for $q$ terminating in
one of the nodes in $G_{i,j}$ then we have, by applying
\thmref{shannon},
\begin{eqnarray*}
   \mu_{D_{i,j}}(T^*) + 2 
    & \ge & \sum_{v\in G_{i,j}}\Pr(v\mid G_{i,j})\log(1/\Pr(v\mid G_{i,j}) \\
    & \ge & \sum_{v\in G_{i,j}}\Pr(v\mid G_{i,j})\log(\Omega(\alpha^i)) \\
    & \ge & \log(\Omega(\alpha^i)) \\
    & = & i\log\alpha - O(1) \enspace .
\end{eqnarray*}

Putting this all together, we obtain
\begin{eqnarray*}
\mu_D(T^*) 
  & = & \sum_{i=1}^{\infty}\sum_{j=1}^{t_i}\Pr(G_{i,j})\mu_{D_{i,j}}(T^*) \\
  & \ge & \sum_{i=1}^{\infty}
    \sum_{j=1}^{t_i-1}\Pr(G_{i,j})\mu_{D_{i,j}}(T^*) \\
  & \ge & \sum_{i=1}^{\infty}\sum_{j=1}^{t_i-1}
           \Pr(G_{i,j})( i\log\alpha -O(1)) \\
  & \ge & (\log\alpha)\cdot
         \sum_{i=1}^{\infty}\sum_{j=1}^{t_i-1}
		\sum_{v\in G_{i,j}}\Pr(v)\cdot\depth(v) - O(1)\\
  & = & (\log\alpha)\cdot\sum_{v\in T}\Pr(v)\cdot \depth(v)
          -    \sum_{i=1}^{\infty}
		\sum_{v\in G_{i,t_i}}\Pr(v)\cdot\depth(v) - O(1)\\
  & \ge & (\log\alpha)\cdot\sum_{v\in T}\Pr(v)\cdot \depth(v)
          -    \sum_{i=1}^{\infty}i\cdot\Pr(G_{i,t_i}) - O(1) \\
  & \ge & (\log\alpha)\cdot\sum_{v\in T}\Pr(v)\cdot \depth(v)
          -    \sum_{i=1}^{\ceil{\log n}}i\gamma^i/2^{i-1} - O(1) \\
  & \ge &  (\log\alpha)\cdot\sum_{v\in T}\Pr(v)\cdot \depth(v) - O(1) \\
  & = & \Omega(C_1(S,D)) \enspace ,
\end{eqnarray*}
where the last inequality follows from the fact that $\gamma/2 < 1$. 
\end{proof}

To get some idea of the constants involved in the proof of
\lemref{depth}, we can select $\gamma=1.6$, so that
$\alpha=1.6/\sqrt{2}\approx 1.13137085$ and $\log \alpha \approx
0.178071905$ and the $O(1)$ term is approximately 20.  Thus, for this
choice of parameters, the depth in $T$ is competitive with $T^*$ to
within a factor of $1/0.178071905\approx 5.615$ and an additive
constant of 20.  Alternatively, selecting $\gamma=1.8$ gives a
constant factor less than 3 and an additive term of approximately 90.

And now the main event:

\begin{thm}
Let $S$ be a set of $n$ points and let $D$ be a probability measure
over $\R^2$.
Let $T=T(S,D)$ be the biased range tree for $S$ and $D$ and 
let $T^*$ be any decision
tree that answers range counting queries for $S$.  Then
\[
  \mu_D(T^*) = \Omega(\mu_D(T)) \enspace .
\]
\end{thm}

\begin{proof}
By the definition of $C_1$ and $C_2$, the expected cost of searching in
$T$ is $\mu_D(T)=O(C_1(S,D)+C_2(S,D))$.  On the other hand, by
\lemref{depth} and \lemref{cataloguer} $\mu_D(T^*) =
\Omega(\max\{C_1(S,D),C_2(S,D)\}) =
\Omega(C_1(S,D)+C_2(S,D))=\Omega(\mu_D(T))$.  This completes the proof.
\end{proof}

\section{Summary, Discussion, and Conclusions}
\seclabel{summary}

We have presented biased range trees, an optimal data structure for
2-sided orthogonal range counting queries when the point set $S$ and
query distribution $D$ is known in advance. The expected time required
to answer queries with a biased range tree, when the queries are
distributed according to $D$, is within a constant factor of any
decision tree for answering range queries over $S$.  Like standard
range trees, biased range trees use $O(n\log n)$ space and can also
answer semigroup queries \cite{ae133,ae292}.\footnote{That biased
range trees can answer semigroup queries follows
from Properties~1--3 of the catalogues in \secref{catalogues}.}
Although the analysis of
biased range trees is complicated, their implementation is not much
more complicated than that of standard range trees.

As a small optimization, the backup range tree data structure can be
eliminated from biased range trees.  Instead, once the probability of
a node $v$ drops below $1/n$ the node can be split by ignoring the
distribution $D$ and simply splitting the points of $r(v)\cap S$ into
two sets of roughly equal size.  This results in a tree of depth at
most $2(\log n+1)$.

This work is just one of many possible results on
distribution-sensitive range searching.  Several open problems
immediately arise.

\begin{op}
Are there efficient distribution-sensitive data structures for 3-sided
and 4-sided orthogonal range counting queries?
\end{op}

Note that a 4-sided orthogonal range counting query can be reduced to
4 2-sided orthogonal range counting queries using the principle of
inclusion-exclusion.  Unfortunately, this reduction does not produce
an optimal distribution-sensitive data structure.  To see this,
consider 4-sided queries consisting of unit squares whose bottom left
corner is uniformly distributed in the shaded region of
\figref{squares}.  All such queries contain no points in the query
region and all such queries can be answered in $O(1)$ time by simply
checking that all four corners of the square are to the left of the point
set.  However, when we decompose these queries into a four 2-sided
queries we obtain 2-sided queries that require $\Omega(\log n)$ time
to be answered.

\begin{figure}
  \begin{center}
    \begin{tabular}{c}
      \includegraphics{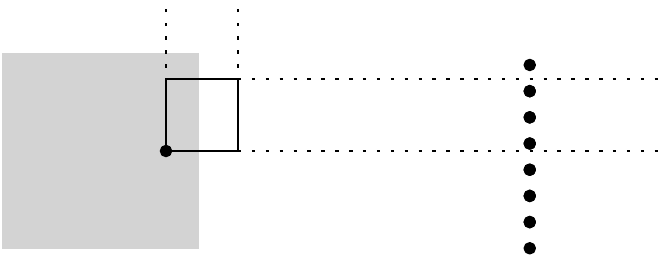}
    \end{tabular}
  \end{center}
  \caption{Decomposing a 4-sided query into four 2-sided queries can
           produce a bad distribution of 2-sided queries.}
  \figlabel{squares}
\end{figure}

\begin{op}
Biased range trees require that the point set $S$ and the
distribution $D$ be known in advance.  Is
there a self-adapting version of biased range trees that, without
knowing $D$ in advance, can answer $m$ queries, each drawn
independently from $D$ in $O(n\log n+ m\mu_D(T^*))$ expected time?
\end{op}

\begin{op}
Determine the worst-case or the average case constants associated with
2-dimensional orthogonal range searching for comparison-based data structures.
By applying the result of Adamy and Seidel \cite{as98} on point
location to the arrangement $A$ described in \secref{preliminaries} one
immediately obtains an $O(n^2)$ space data structure that answers
queries using at most $2\log n + O(\log\log n)$ comparisons.  Is there
an $O(n\log n)$ space structure with the same performance?
\end{op}

\begin{op}
A point $q\in \R^d$ is \emph{maximal} with respect to $S\subseteq\R^d$
if no point of $S$ has every coordinate larger than the corresponding
coordinate of $q$. For $d\ge 3$,  is there a distribution-sensitive 
data structure for
testing if a query point $q$ is maximal?  For point sets in 2
dimensions, an orthogonal variant of the point-location techniques of Collette
\etal\ \cite{cdilm08} seems to apply. 
\end{op}

\begin{op}
Are there distribution-sensitive data structures for $d$-sided range
search in point sets in $\R^d$?  The current fastest structures
for range search in point sets in $\R^d$ that use near-linear space have
$\Theta(\log^{d-1} n)$ query time.  Is there a structure that uses
near-linear space and is optimal
when the point set $S$ and the distribution $D$ are known in advance?
\end{op}

\bibliographystyle{plain}
\bibliography{rangesearch}

\end{document}

%% file: pat.tex
\usepackage{amsthm}

\newcommand{\comment}[1]{}

\newcommand{\seclabel}[1]{\label{sec:#1}}

\newcommand{\secref}[1]{\mbox{Section~\ref{sec:#1}}}

\newcommand{\figlabel}[1]{\label{fig:#1}}

\newcommand{\figref}[1]{\mbox{Figure~\ref{fig:#1}}}

\newcommand{\eqlabel}[1]{\label{eq:#1}}
\newcommand{\eqref}[1]{(\ref{eq:#1})}

\newtheorem{thm}{Theorem}{\bfseries}{\itshape}
\newcommand{\thmlabel}[1]{\label{thm:#1}}
\newcommand{\thmref}[1]{Theorem~\ref{thm:#1}}

\newtheorem{lem}{Lemma}{\bfseries}{\itshape}
\newcommand{\lemlabel}[1]{\label{lem:#1}}
\newcommand{\lemref}[1]{Lemma~\ref{lem:#1}}

{\bfseries}{\itshape}

{\bfseries}{\itshape}

{\bfseries}{\itshape}

\newtheorem{assumption}{Assumption}{\bfseries}{\rm}

\newtheorem{op}{Open Problem}

\newcommand{\etal}{\emph{et al}}

\newcommand{\ceil}[1]{\left\lceil #1 \right\rceil}

\newcommand{\R}{\mathbb{R}}

\usepackage{marvosym}